\documentclass[11pt,reqno]{amsart}
\usepackage{amscd,amssymb,amsmath,amsthm}
\usepackage[arrow,matrix]{xy}
\usepackage{graphicx}
\usepackage{cite}
\usepackage{geometry}
\newtheorem{thm}[subsection]{Theorem}
\newtheorem{lemma}[subsection]{Lemma}
\newtheorem{pro}[subsection]{Proposition}

\newtheorem{rk}[subsection]{Remark}
\newtheorem{defn}[subsection]{Definition}

\numberwithin{equation}{section} 

\newcommand{\s}{{\sigma}}
\newcommand{\de}{{\xi}}

\def \s {\sigma}

\newcommand{\bea}{\begin{eqnarray}}
\newcommand{\eea}{\end{eqnarray}}

\begin{document}

\title[]{Phase transition for models with continuum set of spin values on Bethe lattice}

\author{Yu. Kh. Eshkabilov, G. I. Botirov, F. H. Haydarov}

\address{Yu.\ Kh.\ Eshkabilov\\ National University of Uzbekistan,
Tashkent, Uzbekistan.} \email {yusup62@mail.ru}

\address{G.\ I.\ Botirov\\ Institute of mathematics,
29, Do'rmon Yo'li str., 100125, Tashkent, Uzbekistan.} \email
{botirovg@yandex.ru}

\address{F.\ H.\ Haydarov\\ National University of Uzbekistan,
Tashkent, Uzbekistan.} \email {haydarov\_imc@mail.ru}

\begin{abstract} In this paper we consider models with nearest-neighbor interactions and with the set [0,1] of spin values,
on a Bethe lattice (Cayley tree) of an arbitrary order. These models depend on parameter $\theta$. We describe all of
Gibbs measures in any right parameter $\theta$ corresponding to the models.
\end{abstract}
\maketitle

{\bf Mathematics Subject Classifications (2010).} 82B05, 82B20
(primary); 60K35 (secondary)

{\bf{Key words.}} Cayley tree, spin value, Gibbs measures,
Hammerstein's equation, fixed point.

\section{Introduction} \label{sec:intro}

 Spin models on a graph or in a continuous spaces form a large
class of systems considered in statistical mechanics. Some of them
have a real physical meaning, others have been proposed as
suitably simplified models of more complicated systems. The
geometric structure of the graph or a physical space plays an
important role in such investigations. For example, in order to
study the phase transition problem on a cubic lattice $Z^d$ or in
space one uses, essentially, the Pirogov- Sinai theory; see
\cite{PS1} and \cite{PS2}. A general methodology of phase
transitions in $Z^d$ or $R^d$ was developed in \cite{7}. On the
other hand, on a Cayley tree of order $k$ one uses the theory of
Markov splitting random fields based upon the corresponding
recurrent equations. In particular, in Refs \cite{1}, \cite{3},
\cite{11} and \cite{13} Gibbs measures on $\Gamma_k$ have been
described in terms of solutions to the recurrent equations.

During last five years, an increasing attention was given to
models with a \emph{continuum} set of spin values on a Cayley
tree. Until now, one considered nearest-neighbor interactions
$(J_{3}=J=\alpha=0,\ J_{1}\neq 0)$ with the set of spin values
$[0,1]$. The following results was achieved: splitting Gibbs
measures on a Cayley tree of order $k$ are described by solutions
to a nonlinear integral equation. For $k = 1$ (when the Cayley
tree becomes a one-dimensional lattice $\mathbb{Z}^1$) it has been
shown that the integral equation has a unique solution, implying
that there is a unique Gibbs measure. For a general $k$, a
sufficient condition was found under which a periodic splitting
Gibbs measure is unique (see \cite{ehr2013}, \cite{enh2015},
\cite{rh2015} and \cite{re}).

In \cite{ehr2012} on a Cayley tree $\Gamma_{k}$ of order $k \geq
2$, phase transitions were proven to exist i.e., it was given
examples of Hamiltonian of model which there exists phase
transitions. Afterwards, in \cite{new1} it was generalized the
examples on $\Gamma_{2}$. There are some examples of models with continuum set of spin values which there exists a phase transition on a Cayley tree of some order (see  \cite{t}, \cite{la}, \cite{ehr2012}, \cite{new1}).
In \cite{p} it was considered a model with nearest-neighbor interactions and with the set [0,1] of spin values,
on a Cayley tree of order two. This model depends on two parameters $n\in \mathbb N$ and $\theta\in [0,1)$.  Author proved that if $ 0 \leq \theta \leq \frac{2n+3}{2(2n+1)}$, then for the model there exists a unique translational-invariant Gibbs measure;
If $\frac{2n+3}{2(2n+1)} < \theta <1$, then there are three translational-invariant Gibbs measures (i.e. phase transition occurs).

 In this paper we consider models which
include all of examples in \cite{p}, \cite{ehr2012}, \cite{new1} on a Cayley
tree of an arbitrary order. Also we describe all of Gibbs measures
corresponding to the models.

\section{Preliminaries}
Denote that on the bottom definitions and known results are given
short. The reader can read detail  in \cite{re}.

  A Cayley tree $\Gamma_k=(V,L)$ of order $k\in \mathbb{N}$
is an infinite homogeneous tree, i.e., a graph without cycles,
with exactly $k+1$ edges incident to each vertices. Here $V$ is
the set of vertices and $L$ that of edges (arcs). Two vertices $x$
and $y$ are called nearest neighbors if there exists an edge $l\in
L$ connecting them. We will use the notation $l=\langle
x,y\rangle$. The distance $d(x,y), x,y \in V$, on the Cayley tree
is defined by the formula

$$d(x,y)=\min\{d |\ x=x_{0},x_{1},...,x_{d-1},x_{d}=y\in V \ \emph{such that the pairs}$$
$$\langle x_{0},x_{1}\rangle,...,\langle x_{d-1},x_{d}\rangle \emph{are neighboring vertices}\}.$$\vskip
0.3 truecm

Let $x^{0}\in V$ be fixed and set

$$W_{n}=\{x\in V\ |\ d(x,x^{0})=n\}, \,\,\,\,\ V_{n}=\{x\in V\ |\ d(x,x^{0})\leq n\},$$

$$L_{n}=\{l=\langle x,y\rangle\in L\ |\ x,y \in V_{n}\}.$$\vskip
0.3 truecm
 The set of the direct successors of $x$ is denoted by $S(x),$
 i.e.
 $$S(x)=\{y\in W_{n+1}|\ d(x,y)=1\}, \ x\in W_{n}.$$
 We observe that for any vertex $x\neq x^{0},\ x$ has $k$ direct
 successors and $x^{0}$ has $k+1$. Vertices $x$ and $y$ are called second neighbors, which fact is marked as
 $\rangle x,y\langle,$ if there exist a vertex $z\in V$ such that
 $x$, $z$ and $y$, $z$ are nearest neighbors. We will consider only second neighbors $\rangle x, y \langle,$ for which there
exist $n$ such that $x, y \in W_n$. Three vertices $x,\ y$ and $z$
are called a triple of neighbors in which case we write $\langle
x, y, z\rangle,$ if $\langle x, y \rangle,\ \langle y, z \rangle$
are nearest neighbors and $x,\ z \in
W_n,\ y \in W_{n-1}$, for some $n \in \mathbb{N}$.\\

Consider models where the spin takes values in the set $[0,1]$, and is assigned to the vertexes of the tree. For $A \subset V$ a configuration $\sigma_A$ on $A$ is an arbitrary function $\sigma_A:A\mapsto [0,1].$ Denote $\Omega_A=[0,1]^A$ the set of all configurations on $A$. A configuration $\sigma$ on $V$ is then defined as a function $x \in V \mapsto \sigma(x) \in [0,1]$; the set of all configurations is $[0,1]^V$.

The (formal) Hamiltonian of the model is:
\begin{equation}\label{m}
H(\sigma)=-J \sum \limits_{<x,y> \in L}\xi_{\sigma(x),\sigma(y)},
\end{equation} where $J \in R\setminus \{0\}$ and $\xi : (u,v) \in
[0,1]^2 \mapsto \xi_{u,v} \in R$ is a given bounded, measurable
function.

Let $h:\;x\in V\mapsto h_x=(h_{t,x}, t\in [0,1]) \in R^{[0,1]}$ be
mapping of $x\in V\setminus \{x^0\}$.  Given $n=1,2,\ldots$,
consider the probability distribution $\mu^{(n)}$ on
$\Omega_{V_n}$ defined by
\begin{equation}\label{e2}
\mu^{(n)}(\sigma_n)=Z_n^{-1}\exp\left(-\beta H(\sigma_n)
+\sum_{x\in W_n}h_{\sigma(x),x}\right),
\end{equation}
 Here, as before, $\sigma_n:x\in V_n\mapsto
\sigma(x)$ and $Z_n$ is the corresponding partition function:
\begin{equation}\label{e3} Z_n=\int_{\Omega_{V_n}}
\exp\left(-\beta H({\widetilde\sigma}_n) +\sum_{x\in
W_n}h_{{\widetilde\sigma}(x),x}\right)
\lambda_{V_n}({d\widetilde\s_n}).
\end{equation}

The probability distributions $\mu^{(n)}$ are compatible if for
any $n\geq 1$ and $\sigma_{n-1}\in\Omega_{V_{n-1}}$:
\begin{equation}\label{e4}
\int_{\Omega_{W_n}}\mu^{(n)}(\sigma_{n-1}\vee\omega_n)\lambda_{W_n}(d(\omega_n))=
\mu^{(n-1)}(\sigma_{n-1}).
\end{equation} Here
$\sigma_{n-1}\vee\omega_n\in\Omega_{V_n}$ is the concatenation of
$\sigma_{n-1}$ and $\omega_n$. In this case there exists a unique
measure $\mu$ on $\Omega_V$ such that, for any $n$ and
$\sigma_n\in\Omega_{V_n}$, $\mu \left(\left\{\sigma
\Big|_{V_n}=\sigma_n\right\}\right)=\mu^{(n)}(\sigma_n)$.

\begin{defn} The measure $\mu$ is called {\it splitting
Gibbs measure} corresponding to Hamiltonian (\ref{m}) and
function $x\mapsto h_x$, $x\neq x^0$. \end{defn}

 The following
statement describes conditions on $h_x$ guaranteeing compatibility
of the corresponding distributions $\mu^{(n)}(\sigma_n).$

 \begin{pro}\label{p1}\cite{re} {\it The probability distributions
$\mu^{(n)}(\sigma_n)$, $n=1,2,\ldots$, in} (\ref{e2}) {\sl are
compatible iff for any $x\in V\setminus\{x^0\}$ the following
equation holds:
\begin{equation}\label{e5}
 f(t,x)=\prod_{y\in S(x)}{\int_0^1\exp(J\beta\de_{tu})f(u,y)du \over \int_0^1\exp(J\beta{\de_{0u}})f(u,y)du}.
 \end{equation}
Here, and below  $f(t,x)=\exp(h_{t,x}-h_{0,x}), \ t\in [0,1]$ and
$du=\lambda(du)$ is the Lebesgue measure.}
\end{pro}

\section{Main results}

Let
$$C^+[0,1]=\{f\in C[0,1]: f(x)\geq 0\}.$$

For every $k\in\mathbb{N}$ we consider an integral operator
$H_{k}$ acting in the cone $C^{+}[0,1]$ as
\begin{equation}\label{o}(H_{k}f)(t)=\int^{1}_{0}K(t,u)f^{k}(u)du, \,\ k\in\mathbb{N}.\end{equation}

The operator $H_{k}$ is called Hammerstein's integral operator of
order $k$. This operator is well known to generate ill-posed
problems. Clearly, if $k\geq2$ then $H_{k}$ is a nonlinear
operator.

It is known that the set of translational invariant Gibbs measures
of the model (\ref{m}) is described by the fixed points of the
Hammerstein's operator (see \cite{ehr2013}).

For $k \geq 2$ in the model (\ref{m}) and
$$\xi_{t,u}=\xi_{t,u}(\theta,\beta)=\frac{1}{J \beta}\ln\left(1+\theta \sqrt[2n+1]{4(t-\frac{1}{2})(u-\frac{1}{2})}\right), \ t,u \in [0,1]$$
where $ -4^{\frac{1}{2n+1}}< \theta <4^{\frac{1}{2n+1}}$. The for
the Kernel $K(t,u)$ of the Hammerstein's operator $H_k$ we have
\begin{equation}\label{k}K(t,u)=1+\theta \sqrt[2n+1]{4(t-\frac{1}{2})(u-\frac{1}{2})}.\end{equation}

Let $t-\frac{1}{2}=x$ and $u-\frac{1}{2}=y$, we get
\begin{equation}K(x,y)=1+\theta\sqrt[2n+1]{xy}, \ x, y\in [-0.5,0.5].\end{equation}

We defined the operator $V_2:(x,y) \in R^2 \rightarrow
(x',y') \in R^2$ by
\begin{equation}\label{V}
\begin{split} V_k: \left\{
\begin{array}{lllllll}
x'=\frac{(2n+1)!\cdot k!}{2}\sum \limits_{i=0}^{2n+1}\frac{(-1)^i}{(2n-i)!(k+1+i)!}\cdot \frac{(x+\sqrt[2n+1]{2}\theta y)^{k+1+i}-(-1)^i(x-\sqrt[2n+1]{2}\theta y)^{k+1+i}}{\sqrt[2n+1]{2^i}\theta^i y^i}\\ [6 mm]
y'=\frac{(2n+1)^2(2n)!k!}{2 \sqrt[2n+1]{2}}\sum \limits_{i=0}^{2n+2}\frac{(-1)^i}{(2n+1-i)!(k+1+i)!}\cdot \frac{(x+\sqrt[2n+1]{2}\theta y)^{k+1+i}+(-1)^i(x-\sqrt[2n+1]{2}\theta y)^{k+1+i}}{\sqrt[2n+1]{2^i}\theta^i y^i};
\end{array}\right.
\end{split}
\end{equation}

\begin{lemma} \cite{t}.A function $\varphi \in C[0,1]$ is a solution of the Hammerstein's equation
\begin{equation}
\begin{split}
(H_kf)(t)=f(t)\end{split}
\end{equation}
iff $\varphi(t)$ has the following form
$$\varphi(t)=C_1 + C_2
\theta \sqrt[2n+1]{4(t-\frac{1}{2})},$$
where $(C_1, C_2) \in R^2$ is
a fixed point of the operator $V_k$ (\ref{V}).
\end{lemma}

 For $k= 2s, \ s \in \mathbb{N}$, we denote following notations $$\alpha_{2i}=\frac{2n+1}{2n+2i+1}\left(\frac{1}{2}\right)^{\frac{2i}{2n+1}}, \ \ \ \beta_{i}=C_{2s}^{i} \alpha_{i+1}, \ i, n \in \mathbb{N};$$
$$\theta_{2i+1}=\frac{(2i+1)(2n+2i+3)4^{\frac{1}{2n+1}}}{(2s-2i)(2n+2i+1)}.$$

For $k= 2s+1, \ s \in \mathbb{N}$, we denote following notations $$\alpha_{2i}=\frac{2n+1}{2n+2i+1}\left(\frac{1}{2}\right)^{\frac{2i}{2n+1}}, \ \ \ \beta_{i}=C_{2s+1}^{i+1} \alpha_{i+1}, \ i, n \in \mathbb{N};$$
$$\theta_{2i+1}=\frac{(2i+1)(2s+2i+3)4^{\frac{1}{2n+1}}}{(2s-2i+1)(2n+2i+1)}.$$

  \begin{rk} We consider the following function $\theta_{x}=\frac{x^2+(2n+2)x}{(2s+1-x)(x+2n)}$.
  From $\theta^{'}_{x}=\frac{2x^{2}+(2s+2)(x+2n)^2+4n(2s+1)}{(2s+1-x)^{2}(x+2n)^{2}}>0$ we can conclude that
\begin{equation}\label{t}\theta_{1}<\theta_{3}<\theta_{5}<...<\theta_{2s-1}.\end{equation}\end{rk}

\begin{lemma} Let $k=2s,\ s\in \textbf{N}$. If the point $\gamma (x_0,y_0) \in R^+_2$ is a fixed point of (\ref{V}), then $\gamma \in \mathbb{R^>_2} $ and $\lambda=\frac{y}{x}$ is a root of the following equation
\begin{equation}\label{p} P(\lambda):=\beta_{1}(\theta_1-\theta)+\beta_{3}(\theta_{3}-\theta)\lambda^{2}+
...+\beta_{2s-1}(\theta_{2s-1}-\theta)\lambda^{2s-2}+\alpha_{2s}\lambda^{2s}=0.
\end{equation}
 \end{lemma}

\begin{proof}  Let $(x_0, y_0)$ is a fixed point of (\ref{V}). Now, we divide the second part to first part of system (\ref{V}) then we get following \begin{equation}\frac{\lambda}{\theta}=\frac{C_{2s}^{1}
\alpha_{2}\lambda+C_{2s}^{2} \alpha_{3}\lambda^{2}+...+C_{2s}^{2s}
\alpha_{2s+1}\lambda^{2s}}{1+C_{2s}^{1}
\alpha_{1}\lambda+C_{2s}^{2} \alpha_{2}\lambda^{2}+...+C_{2s}^{2s}
\alpha_{2s}\lambda^{2s}} \end{equation} where  $\lambda=\frac{y}{x}$.

Let $\lambda \neq 0$. After some abbreviations we get

\begin{equation}1-C_{2s}^{1}
\alpha_{2}\theta+(C_{2s}^{2} \alpha_{2}-C_{2s}^{3}
\alpha_{4}\theta)\lambda^{2}+...+(C_{2s}^{2s-2}
\alpha_{2s-2}-C_{2s}^{2s-1}
\alpha_{2s}\theta)\lambda^{2s-2}+C_{2s}^{2s}
\alpha_{2s}\lambda^{2s}=0.\end{equation} Namely,
\begin{equation}\beta_{1}(\theta_1-\theta)+\beta_{3}(\theta_{3}-\theta)\lambda^{2}+
...+\beta_{2s-1}(\theta_{2s-1}-\theta)\lambda^{2s-2}+\alpha_{2s}\lambda^{2s}=0,\end{equation}
where $\theta_{2i-1}=\frac{C_{2s}^{2i-2}
\alpha_{2i-2}}{C_{2s}^{2i-1} \alpha_{2i}}$. It is easy to see if
$\lambda=0$ then this solution corresponding to solution $(1,0)$
of (\ref{V}).\end{proof}

Analogously, we get the following Lemma
\begin{lemma} Let $k=2s+1,\ s\in \textbf{N}.$ If the point $\gamma (x_0,y_0) \in R^+_2$ is a fixed point of (\ref{V}), then $\gamma \in \mathbb{R^>_2} $ and $\lambda=\frac{y}{x}$ is a root of the following equation
\begin{equation}\label{p1} Q(\lambda):=\beta_{1}(\theta_1-\theta)+\beta_{3}(\theta_{3}-\theta)\lambda^{2}+
...+\beta_{2s-1}(\theta_{2s-1}-\theta)\lambda^{2s-2}+\beta_{2s+1}(\theta_{2s+1}-\theta)\lambda^{2s}=0.
\end{equation}
 \end{lemma}

\begin{pro}\label{pr1} Let $k=2s,\ s\in \textbf{N}.$\\
a) If $\theta\leq \theta_{1}$, then there is no non-trivial solution of (\ref{p});\\
b) If $\theta>\theta_{1}$, then there is exactly two (non-trivial) solutions of (\ref{p}). These solutions are opposing.
\end{pro}
\begin{proof} Proof Case a) of the Proposition is clearly. \\
b) Number of sign changes of coefficients of $P(\lambda)$ is equal
to 1. Then $P(\lambda)$ has at most one positive solution. The
second hand side we have $P(0)<0$ and
$\lim_{\lambda\rightarrow\infty}P(\lambda)=+\infty$. Then by
Roll's theorem $P(\lambda)$ has at least one positive solution.
Thus, there exist $\lambda^{\ast}>0$ such that
$P(\lambda^{\ast})=0$. Since $P(\lambda)$ is an even function
there is only one negative solution, i.e., $-\lambda^{\ast}$.\\
\end{proof}

\begin{pro}\label{pr2} Let $k=2s+1,\ s\in \textbf{N}.$\\
a) If $\theta\leq \theta_{1}$, then there is no non-trivial solution of (\ref{p1});\\
b) If $\theta>\theta_{1}$, then there is exactly two (non-trivial)
solutions of (\ref{p1}). These solutions are opposing.
\end{pro}

\begin{proof} Proof of Proposition \ref{pr2} is similar to proof
of Proposition \ref{pr1}
\end{proof}

\begin{pro}\label{th} Let $k=2s,\ s\in \textbf{N}.$\\
a) Let $-4^{\frac{1}{2n+1}}<\theta\leq \theta_{1}$. Then
(\ref{o}) has only one
positive fixed point: $f(t)=1$.\\
b)Let
$$\frac{\sum_{i=1}^{s}2^{\frac{2i-2}{2n+1}}\beta_{2i-1}\theta_{2i-1}+\alpha_{2s}2^{\frac{2s}{2n+1}}}
{\sum_{i=1}^{s}2^{\frac{2i-2}{2n+1}}\beta_{2i-1}}\leq \theta\leq
4^{\frac{1}{2n+1}}.$$ Then (\ref{o}) has exactly two positive
fixed points: $f_{1}(t)=1$,
$f_{2}(t)=\bar{C}(1+\lambda^{\ast}t^{\frac{1}{2n+1}})$, where $\lambda^{\ast}$ is a positive solution (\ref{p}).\\
 c)Let
$$\theta_{1}<\theta<
\frac{\sum_{i=1}^{s}2^{\frac{2i-2}{2n+1}}\beta_{2i-1}\theta_{2i-1}+\alpha_{2s}2^{\frac{2s}{2n+1}}}
{\sum_{i=1}^{s}2^{\frac{2i-2}{2n+1}}\beta_{2i-1}}.$$ Then
(\ref{o}) has exactly three positive fixed points: $f_{1}(t)=1$,
$f_{2}(t)=\bar{C}(1+\lambda^{\ast}t^{\frac{1}{2n+1}})$,
$f_{3}(t)=\bar{C}(1-\lambda^{\ast}t^{\frac{1}{2n+1}})$ , where
$\lambda^{\ast}$ is a positive solution (\ref{p}).
\end{pro}

\begin{proof} We'll prove that case b (case a and c are similarly).
From $$\theta\geq
\frac{\sum_{i=1}^{s}2^{\frac{2i-2}{2n+1}}\beta_{2i-1}\theta_{2i-1}+\alpha_{2s}2^{\frac{2s}{2n+1}}}
{\sum_{i=1}^{s}2^{\frac{2i-2}{2n+1}}\beta_{2i-1}}>\theta_{1}$$
 (\ref{p}) has exactly 3 solutions. They are
$\lambda_{1}=0$, $\lambda_{2}=\lambda^{\ast}$ and
$\lambda_{3}=-\lambda^{\ast}$.
 By definition of $f(t)$ we get following solutions: $f_{1}(t)=1$,
$f_{2}(t)=C_{1}(1+\lambda^{\ast}t^{\frac{1}{2n+1}})$ and
$f_{3}(t)=C_{1}(1-\lambda^{\ast}t^{\frac{1}{2n+1}})$. But it is
interesting for us to find positive solutions, that's why we need
positive solutions. It's easy to check that $f_{1}(t), \ f_{2}(t)$
are positive solutions. We must check the third solution. The
third solution $f_{3}(t)$ be a negative if and only if
$\lambda^{\ast}\geq 2^{\frac{1}{2n+1}}$. Namely, it's sufficient
to check that $P(2^{\frac{1}{2n+1}})<0$. The last inequality is
equivalent to
$$ \theta\geq
\frac{\sum_{i=1}^{s}2^{\frac{2i-2}{2n+1}}\beta_{2i-1}\theta_{2i-1}+\alpha_{2s}2^{\frac{2s}{2n+1}}}
{\sum_{i=1}^{s}2^{\frac{2i-2}{2n+1}}\beta_{2i-1}}.$$ \end{proof}
Thus we have proved the following
\begin{thm} Let $k=2s,\ s\in \textbf{N}.$

(a) If $-4^{\frac{1}{2n+1}}<\theta\leq \theta_{1}$, then for model (\ref{m}) on the Cayley tree of order $k$ there exists the unique translation-invariant Gibbs measure;

(b) If $$\frac{\sum_{i=1}^{s}2^{\frac{2i-2}{2n+1}}\beta_{2i-1}\theta_{2i-1}+\alpha_{2s}2^{\frac{2s}{2n+1}}}
{\sum_{i=1}^{s}2^{\frac{2i-2}{2n+1}}\beta_{2i-1}}\leq \theta\leq
4^{\frac{1}{2n+1}},$$ then for model (\ref{m}) on the Cayley tree of order $k$ there are exactly two translation-invariant Gibbs measures;

(c) If $$\theta_{1}<\theta<
\frac{\sum_{i=1}^{s}2^{\frac{2i-2}{2n+1}}\beta_{2i-1}\theta_{2i-1}+\alpha_{2s}2^{\frac{2s}{2n+1}}}
{\sum_{i=1}^{s}2^{\frac{2i-2}{2n+1}}\beta_{2i-1}},$$ then for model (\ref{m}) on the Cayley tree of order $k$ there are exactly three translation-invariant Gibbs measures.
\end{thm}

Similar to Proposition \ref{th}, we get the following

\begin{pro}\label{th1} Let $k=2s+1,\ s\in \textbf{N}.$\\
a) Let $-4^{\frac{1}{2n+1}}<\theta\leq \theta_{1}$,
$\theta_{2s+1}\leq\theta<4^{\frac{1}{2n+1}}$. Then (\ref{o}) has
only one
positive fixed point: $f(t)=1$.\\
b)Let
$$\frac{\sum_{i=1}^{s+1}2^{\frac{2i-2}{2n+1}}\beta_{2i-1}\theta_{2i-1}}
{\sum_{i=1}^{s+1}2^{\frac{2i-2}{2n+1}}\beta_{2i-1}}\leq
\theta<\theta_{2s+1}.$$ Then (\ref{o}) has exactly two positive
fixed points: $f_{1}(t)=1$,
$f_{2}(t)=\bar{C}(1+\lambda^{\ast}t^{\frac{1}{2n+1}})$, where $\lambda^{\ast}$ is a positive solution (\ref{p1}).\\
 c)Let
$$\theta_{1}<\theta<
\frac{\sum_{i=1}^{s+1}2^{\frac{2i-2}{2n+1}}\beta_{2i-1}\theta_{2i-1}}
{\sum_{i=1}^{s+1}2^{\frac{2i-2}{2n+1}}\beta_{2i-1}}.$$ Then
(\ref{o}) has exactly three positive fixed points: $f_{1}(t)=1$,
$f_{2}(t)=\bar{C}(1+\lambda^{\ast}t^{\frac{1}{2n+1}})$,
$f_{3}(t)=\bar{C}(1-\lambda^{\ast}t^{\frac{1}{2n+1}})$ , where
$\lambda^{\ast}$ is a positive solution (\ref{p1}).
\end{pro}
Thus we obtain the following
\begin{thm} Let $k=2s+1,\ s\in \textbf{N}.$

(a) If $-4^{\frac{1}{2n+1}}<\theta\leq \theta_{1}$,
$\theta_{2s+1}\leq\theta<4^{\frac{1}{2n+1}}$, then for model (\ref{m}) on the Cayley tree of order $k$ there exists the unique translation-invariant Gibbs measure;

(b) If $$\frac{\sum_{i=1}^{s+1}2^{\frac{2i-2}{2n+1}}\beta_{2i-1}\theta_{2i-1}}
{\sum_{i=1}^{s+1}2^{\frac{2i-2}{2n+1}}\beta_{2i-1}}\leq
\theta<\theta_{2s+1},$$  then for model (\ref{m}) on the Cayley tree of order $k$ there are exactly two translation-invariant Gibbs measures;

(c) If $$\theta_{1}<\theta<
\frac{\sum_{i=1}^{s+1}2^{\frac{2i-2}{2n+1}}\beta_{2i-1}\theta_{2i-1}}
{\sum_{i=1}^{s+1}2^{\frac{2i-2}{2n+1}}\beta_{2i-1}},$$ then for model (\ref{m}) on the Cayley tree of order $k$ there are exactly three translation-invariant Gibbs measures.
\end{thm}

\begin{rk} a) For the case $k=2$ Theorem 3.8 coincides with Theorem 4.2 in \cite{t};\\
b) For the case $k=3$ Theorem 3.10 coincides with Theorem 5.2 in \cite{t}.
\end{rk}

\end{document}